\documentclass[a4paper]{article}

\usepackage{graphicx}
\usepackage{wrapfig}
\usepackage{color}
\usepackage{times}
\usepackage{verbatim}
\usepackage{amsmath,amsfonts,amsthm}
\usepackage{enumerate,url}
\usepackage{wasysym}
\usepackage[labelfont=bf,size=small]{caption}
\usepackage[labelfont=normalfont]{subfig}
\usepackage{tabularx,booktabs}
\usepackage{todonotes}
\usepackage{hyperref}
\usepackage{xspace}
\usepackage{timetravel}
\usepackage{fullpage}

\setCounters{theorem}

\title{Semantic Word Cloud Representations:\\
Hardness and Approximation Algorithms}

\author{%
Lukas~Barth\footnote{Institute of Theoretical Informatics, Karlsruhe
  Institute of Technology} \and
Sara~Irina~Fabrikant\footnote{Department of Geography, University of
  Zurich} \and 
Stephen~Kobourov\footnote{Department of Computer Science, University
  of Arizona} \and
Anna~Lubiw\footnote{School of Computer Science, University of Waterloo} \and
Martin~N\"ollenburg\footnotemark[1] \and
Yoshio~Okamoto\footnote{Dept.~Comm.~Engineering and Informatics,
  University of Electro-Communications} \and 
Sergey~Pupyrev\footnotemark[3] \and
Claudio Squarcella\footnote{Dipartimento di Ingegneria, Roma Tre
  University} \and 
Torsten~Ueckerdt\footnote{Department of Mathematics, Karlsruhe
  Institute of Technology} \and 
Alexander~Wolff\footnote{Lehrstuhl f\"ur Informatik I, Universit\"at
  W\"urzburg}
}

\date{}

\newcommand{\rotate}[1]{{#1}}  

\newcommand{\extremal}[1]{{#1}}  

\graphicspath{{figures/}}

\newcommand{\tbcr}{Contact Representation of Word Networks\xspace}
\newcommand{\fbcr}{\textsc{CROWN}\xspace}
\newcommand{\fbcropt}{\textsc{Max}-\fbcr}
\newcommand{\fbcrhier}{\textsc{Hier}-\fbcr}
\newcommand{\fbcrarea}{\textsc{Area}-\fbcr}
\newcommand{\NP}{{\cal NP}\xspace}
\newcommand{\prob}[1]{\textsc{#1}}
\newcommand{\N}{\ensuremath{\mathbb{N}}\xspace}
\newcommand{\Wopt}{\ensuremath{W_\mathrm{\!opt}}\xspace}

\theoremstyle{plain}

\newtheorem{theorem}{Theorem} 
\newtheorem{lemma}{Lemma}

\newtheorem{corollary}{Corollary}

\begin{document}

\maketitle

\begin{abstract}
We study a geometric representation problem, where we are given a set~$\cal R$ of axis-aligned rectangles
with fixed dimensions and a graph with vertex set~$\cal R$. The task is to place the rectangles
without overlap such that two rectangles touch if and only if the
graph contains an edge between them. We call this
problem \prob{\tbcr} (\fbcr). It formalizes
the geometric problem behind drawing word clouds
in which semantically related words are close to each other. Here, we represent words by rectangles and semantic relationships by edges.

We show that \fbcr is strongly NP-hard even restricted trees and
weakly NP-hard if restricted stars.  We consider the optimization problem \fbcropt where each adjacency induces a certain profit
and the task is to maximize the sum of the profits.  For this problem, we present constant-factor approximations for several
graph classes, namely stars, trees, planar graphs, and graphs of bounded degree. Finally, we evaluate the algorithms experimentally and
show that our best method improves upon the best existing heuristic by 45\%.
\end{abstract}


\section{Introduction}\label{sec:intro}

Word clouds and tag clouds are popular tools for visualizing text. The
practical tool, Wordle~\cite{wordle09}, took word clouds to the next
level with high quality design, graphics, style and
functionality. Such word cloud visualizations provide an appealing way
to summarize the content of a webpage, a research paper, or a
political speech. Often such visualizations are used to contrast two
documents; for example, word cloud visualizations of the speeches
given by the candidates in the 2008 US Presidential elections were
used to draw sharp contrast between them in the popular media.

While some of the more recent word cloud visualization tools aim to
incorporate semantics in the layout, none provides any guarantees
about the quality of the layout in terms of semantics. We propose a
mathematical model of the problem, via a simple 
edge-weighted graph. The vertices in the graph are the words in the
document.
The edges in the graph correspond to semantic relatedness,
with weights corresponding to the strength of the relation. Each
vertex must be drawn as an axis-aligned rectangle (\emph{box}, for
short) with fixed dimensions.  Usually, the dimensions will be
determined by the size of the word in a certain font, and the
font size will be related to the importance of the word.
The goal is to ``realize'' as many edges as possible, by
contacts between their corresponding rectangles; see
Fig.~\ref{fig:complexity-classes}.

\subsection{Related Work.}

Hierarchically clustered document collections are visualized with
self-organizing maps~\cite{HKK96} and Voronoi
treemaps~\cite{brandes12}. The early word-cloud approaches did not
explicitly use semantic information, such as word relatedness, in
placing the words in the cloud. More recent approaches attempt to do
so, as in ManiWordle~\cite{maniwordle} and in parallel tag
clouds~\cite{collins-09}. The most relevant approaches rely on
force-directed graph visualization methods~\cite{Cui_2010_wordcloud}
and a seam-carving image processing method together with a
force-directed heuristic~{\em et al.}~\cite{wu2011semantic}.
The semantics-preserving word cloud problem is related to classic
graph layout problems, where the goal is to draw graphs so that vertex
labels are readable and Euclidean distances between pairs of vertices
are proportional to the underlying graph distance between
them. Typically, however, vertices are treated as points and label
overlap removal is a post-processing step~\cite{dwyer05,gh10}.

In {\em rectangle representations} of graphs, vertices are
axis-aligned rectangles with non-intersecting interiors and edges
correspond to rectangles with non-zero length common boundary. Every
graph that can be represented this way is planar and every triangle in
such a graph is a facial triangle; these two conditions are also
sufficient to guarantee a rectangle
representation~\cite{thomassen1986interval,rosenstiehl1986rectilinear,buchsbaum08,fusy2009transversal}. 
In a recent survey, Felsner~\cite{felsner2013rectangle} reviews many
rectangulation variants, including squarings.
Algorithms for area-preserving rectangular cartograms are also
related~\cite{r-rsc-34}. Area-universal rectangular representations
where vertex weights are represented by area have been
characterized~\cite{eppstein2012area} and edge-universal
representations, where edge weights are represented by length of
contacts have been studied~\cite{nollenburg2013edge}.  Unlike
cartograms, in our setting there is no inherent geography, and hence,
words can be positioned anywhere. Moreover, each word has fixed
dimensions enforced by its frequency in the input text, rather than
just fixed area.

\begin{figure}[tb]
  \centering
  \includegraphics[width=.5\textwidth]{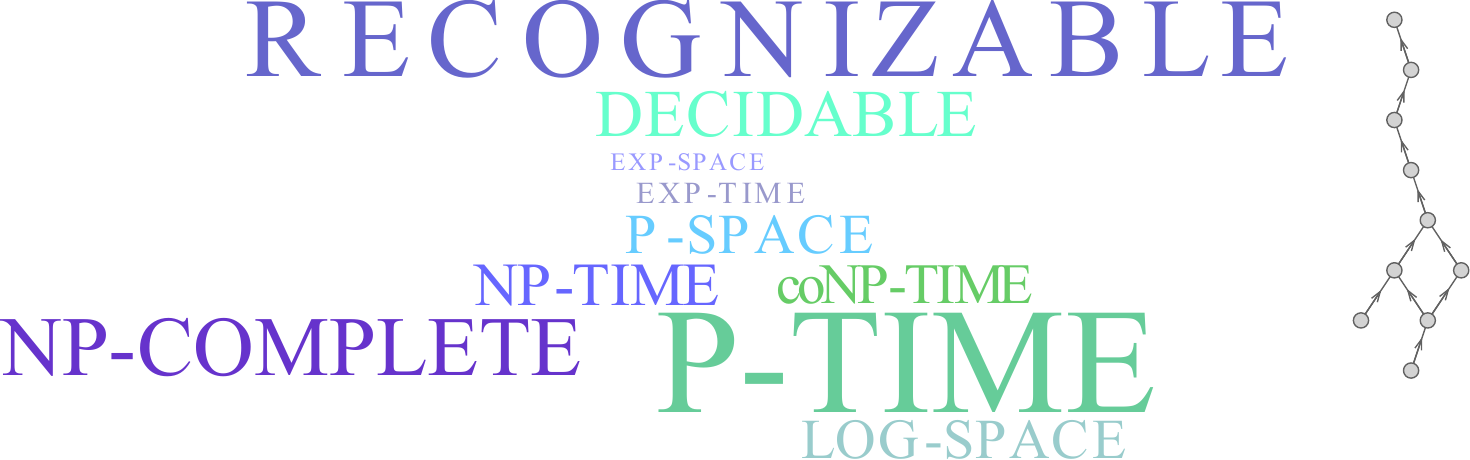}
  \caption{A hierarchical word cloud for complexity classes. A class
    is above another class when the former contains the latter. The
    font size is the square root of millions of Google hits for the
    corresponding word. This is an instance of the problem variant
    \fbcrhier.}
  \label{fig:complexity-classes}
\end{figure}

\subsection{Our Contribution.}

The input to the problem variants that we consider is a sequence
$B_1,\ldots,B_n$ of axis-aligned boxes with fixed positive dimensions.
Box $B_i$ is encoded by $(w_i,h_i)$, where $w_i$ and $h_i$ are its
width and height.  \rotate{For some of our results, some boxes may be
  rotated by $90^\circ$, which means exchanging $w_i$ and $h_i$.}
A \emph{representation} of the boxes $B_1,\ldots,B_n$ is a map that
associates with each box a position in the plane so that no two
boxes overlap.  A \emph{contact} between two boxes is a line segment
(possibly a point) in the boundary of both.  If two boxes are in
contact, we say that they \emph{touch}.  If two boxes touch and one
lies above the other, we call this a \emph{vertical contact}.  We
define \emph{horizontal contact} symmetrically.  For $1 \le i\neq j
\le n$, a non-negative \emph{profit}~$p_{ij}$ represents the gain for
making boxes $B_i$ and $B_j$ touch.  The \emph{supporting graph} has a
vertex for each box and an edge for each non-zero profit.  Finally, we
define the \emph{total profit} of a representation to be the sum of
profits over all pairs of touching boxes.

Our problems and results are as follows.

{\bf\tbcr (\fbcr):} In this decision problem,
we assume 0--1 profits.  The task is to decide whether there
exists a representation of the boxes with total profit $\sum_{i\neq
  j}p_{ij}$. This is equivalent to finding a representation whose
contact graph contains the supporting graph as a subgraph. If such a
representation exists, we say that it \emph{realizes the supporting
  graph} and that the instance of the \fbcr{} problem is
\emph{realizable}. We show that \fbcr is strongly \NP-hard
even if restricted to trees and weakly \NP-hard if restricted stars;
see Theorem~\ref{thm:trees:hardness}.

We also consider two variants of the problem that can be solved
efficiently. First we present a linear-time algorithm for \fbcr on
so-called irreducible triangulations; see
Section~\ref{sec:triangulation}.  Then we turn to the problem variant
\fbcrhier, where the supporting graph is a single-source directed
acyclic graph with fixed plane embedding, and the task is to find a
representation in which each edge corresponds to a vertical contact
directed upwards; see Fig.~\ref{fig:complexity-classes}.  We solve
this variant efficiently; see Section~\ref{sec:hierarchy}.


{\bf\fbcropt:} In this
optimization problem, the task is to find a representation of the
given boxes maximizing the total profit.  We present constant-factor
approximation algorithms for stars, trees, and planar graphs, and a
$2/(\Delta+1)$-approximation for graphs of maximum
degree~$\Delta$; see Section~\ref{sec:optimize}.
We have implemented two approximation algorithms and
evaluated them experimentally in comparison to three existing
algorithms (two of which semantics-aware). Based on a dataset of 120
Wikipedia documents our best method outperforms the best previous
methods by more than 45\%; see Section~\ref{sec:experimental}.
\extremal{We also consider an extremal
  version of the \fbcropt{} problem and show that if the supporting
  graph is $K_n$ ($n \geq 5$) and each profit is $1$, then there
  always exists a representation with total profit $2n-2$ and that
  this is sometimes the best possible. Such a representation can be
  found in linear time. 
}

{\bf \fbcrarea} is 
as follows: Given a realizable instance of \fbcr, find a
representation that realizes the supporting graph and minimizes the
area of a box containing all input boxes.  We show that this
problem is \NP-hard even if restricted to paths; see
Section~\ref{sec:area}.

\section{The \fbcr{} problem}\label{sec:realize}

In this section, we investigate the complexity of \fbcr for several
graph classes.




\begin{theorem}\label{thm:trees:hardness}
  \fbcr is (strongly) \NP-hard.  The problem remains strongly \NP-hard
  even if restricted to trees and weakly \NP-hard if restricted to
  stars.
\end{theorem}

\begin{proof}
  To show that \fbcr on stars is weakly \NP-hard,
  we reduce from the weakly \NP-hard problem \prob{Partition}, which
  asks whether a given multiset of $n$ positive integers
  $a_1, \ldots, a_n$ that sum to~$B$ can be partitioned
  into two subsets, each of sum $B/2$.  We construct a star graph whose
  central vertex corresponds to an $(B/2,\delta)$-box (for some $0 <
  \delta < \min_i a_i$).  We add four leaves corresponding to
  $(B,B)$-squares and, for $i=1,\dots,n$, a leaf corresponding to an
  $(a_i,a_i)$-square.  It is easy to verify that there is a
  realization for this instance of \fbcr if and only if the set can be
  partitioned.

  To show that \fbcr is (strongly)
  \NP-hard, we reduce from \prob{3-Partition}: Given a
  multiset~$S$ of $n = 3m$ integers with $\sum S = mB$, is there a
  partition of~$S$ into $m$ subsets $S_1,\ldots,S_m$ such that $\sum
  S_1 = \dots = \sum S_m=B$?  It is known that \prob{3-Partition} is
  \NP-hard even if, for every $s \in S$, we have $B/4 < s < B/2$,
  which implies that each of the subsets $S_1,\dots,S_m$ must contain
  exactly three elements~\cite{npproofs}.

  Given an instance $S = \{s_1,s_2,\ldots,s_n\}$ of \prob{3-Partition}
  as described above, we define a tree $T_S$ on $n + 4(m-1) +7$
  vertices as in Fig.~\ref{fig:tree:hardness} (for $n=9$ and $m=3$).
  Let $K = (m+1)B + m+1$.  We make a vertex~$c$ of size $(K,1/2)$.
  For each $i=1,\dots,n$, we make a vertex~$v_i$ of size $(s_i, B)$.
  For each $j=0,\dots,m$, we make vertices~$u_j$ and~$b_j$ of size
  $(1,B)$ and vertices $\ell_j$ and $r_j$ of size $(B/2,B)$.  Finally,
  we make vertices $a_1,\dots,a_5$ of size $(K,K)$, and vertices $d_1$
  and $d_2$ of size $(B/2,B)$.
  The tree $T_S$ is as shown by the thick lines in
  Fig.~\ref{fig:tree:hardness}: vertex~$c$ is adjacent to all the
  $v_i$'s, $u_j$'s, $a$'s, and $d$'s; and each vertex $u_j$ is
  adjacent to $b_j$, $\ell_j$, and $r_j$.

\begin{figure}[tb]
  \centering
  \includegraphics[width=4in]{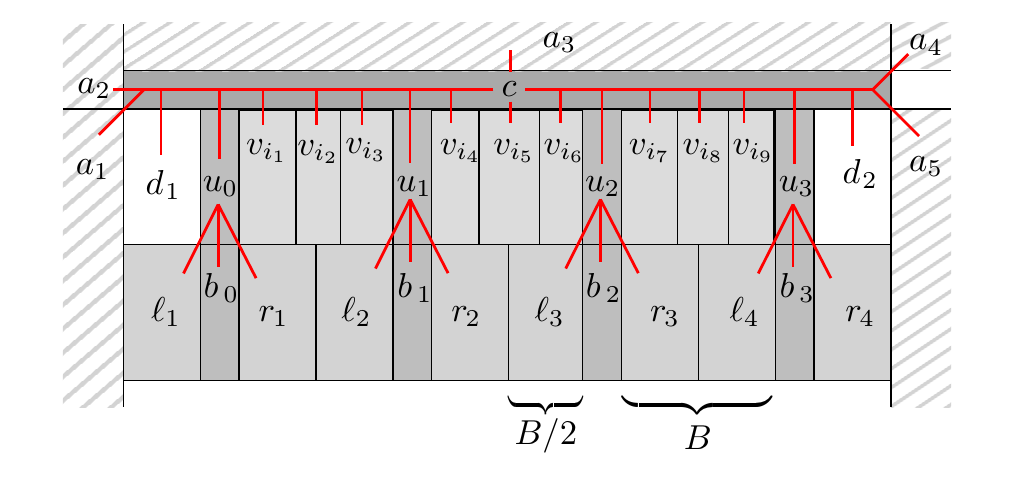}
  \caption{Given an instance~$S$ of \prob{3-Partition}, we construct a
    tree~$T_S$ (thick red line segments) and define boxes such that
    $T_S$ has a realization if and only if $S$ is feasible.}
  \label{fig:tree:hardness}
\end{figure}

  We claim that an instance $S$ of \prob{3-Partition} is feasible if
  and only if $T_S$ can be realized with the given box sizes.  It is
  easy to see that~$T_S$ can be realized if~$S$ is feasible: we simply
  partition vertices $v_1,\dots,v_n$ into groups of three (by vertices
  $u_0,\dots,u_m$) in the same way as their widths $s_1,\dots,s_n$ are
  partitioned in groups of three; see Fig.~\ref{fig:tree:hardness}.

  For the other direction, consider any realization of~$T_S$.  By abusing
  notation, we refer to the box of some vertex~$v$ also as~$v$.
  Since~$c$ touches the five large squares
  $a_1,\dots,a_5$, 
  at least three sides of~$c$ are partially covered by some~$a_k$ and
  at least one horizontal side of~$c$ is completely covered by
  some~$a_k$.  Since~$c$ has height~1/2 only, but touches all the
  $v_i$'s and $u_j$'s and $d_1$ and $d_2$ (each of height $B>1$), all
  these boxes must
  touch~$c$ on its free horizontal side, say, the bottom side.
  Furthermore, the sum of the widths of the boxes exactly matches the
  width of $c$; so they must pack side by side in some order.

  This means that the only free boundary of $u_j$ is at the bottom,
  and $u_j$ must make contact there with $b_j$, $\ell_j$, and $r_j$.
  This is only possible if $b_j$ is placed directly beneath $u_j$, and
  $\ell_j$ and $r_j$ make contact with the bottom corners of $u_j$.
  (They need not appear to the left and right as shown in
  Fig.~\ref{fig:tree:hardness}.)  Because the sum of the widths of the
  $b_j$'s, $\ell_j$'s, and $r_j$'s exactly matches the width of $c$,
  they must pack side by side, and therefore the $u_j$'s are spaced
  distance $B$ apart.  There is a gap of width $B/2$ before the
  first $u_j$ and after the last $u_j$.  These gaps are too wide for
  one box in $v_1,\ldots,v_n$ and too small for two of them since
  their widths are contained in the \emph{open} interval $(B/4,B/2)$.
  Therefore, the boxes $d_1$ and $d_2$ must occupy these gaps, and the
  boxes $v_1,\ldots,v_n$ are packed into $m$ groups each of width~$B$,
  as required.
\end{proof}
\rotate{Note that the proof of the weak \NP-hardness for stars still works in
case rectangles may be rotated because all boxes are squares---but
one.  The same holds for the strong \NP-hardness for trees;
  for details see Appendix~\ref{sec:fbcr}.}

Although \fbcr is \NP-hard in general, there are graph classes for
which the problem can be solved efficiently.  In the remainder of this
section, we investigate such a class---irreducible
triangulations---, and we consider a restricted variant of \fbcr:
\fbcrhier.

\subsection{The \fbcr{} problem on irreducible triangulations}
\label{sec:triangulation}

A box representation is called a \emph{rectangular dual} if the
union of all rectangles is again a rectangle whose boundary is formed
by exactly four rectangles. A graph $G$ admits a rectangular dual if
and only if $G$ is planar, internally triangulated, has a quadrangular
outer face and does not contain separating
triangles~\cite{buchsbaum08}. Such graphs are known as
\emph{irreducible triangulations}. The four outer vertices of an
irreducible triangulation are denoted by $v_N$, $v_E$, $v_S$, $v_W$ in
clockwise order around the outer quadrangle. An irreducible
triangulation $G$ may have exponentially many rectangular
duals.  Any rectangular dual of~$G$, however, can be built up by
placing one rectangle at a time, always keeping the union of the placed
rectangles in staircase shape. 

\wormhole{quasi-triangulated}
\begin{theorem}
  \label{thm:quasi-triangulated}
  \fbcr on irreducible triangulations can be solved in linear time.
\end{theorem}
\begin{proof}[sketch]
  The algorithm greedily builds up the supporting graph $G$, similarly
  to an algorithm for edge-proportional rectangular
  duals~\cite{nollenburg2013edge}.
  We define \emph{concavity} as a point on the boundary of the so-far
  constructed representation, which is a bottom-right or top-left
  corner of some rectangle. Start with a vertical and a
  horizontal ray emerging from the same point $p$, as placeholders for
  the right side of $v_W$ and the top side of $v_S$,
  respectively. Then at each step consider a concavity, with
  $p$ as the initial one. Since each concavity $p$ is contained
  in exactly two rectangles, there exists a unique rectangle $R_p$
  that is yet to be placed and has to touch both these rectangles. If
  by adding $R_p$ we still have a staircase shape representation, then
  we do so. If no such rectangle can be added, we conclude that $G$ is
  not realizable; see Fig.~\ref{fig:quasi-triangulated}. The complete
  proof is in the appendix.
\end{proof}

\begin{figure}[tb]
  \subfloat{\includegraphics{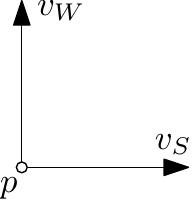}}
  \hfill
  \subfloat{\includegraphics{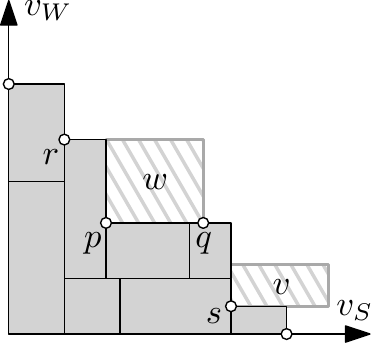}}
  \hfill
  \subfloat{\includegraphics{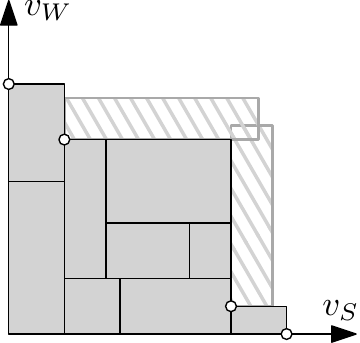}}

  \caption{Left: starting configuration with rays $v_S$ and $v_W$.
    Center: representation at an intermediate step: vertex $w$ fits
    into concavity $p$ and results in a staircase, vertex $v$ fits
    into concavity $s$ but does not result in a staircase.  Adding
    box~$w$ to the representation introduces a new concavity~$q$ and
    allows wider boxes to be placed at~$r$.  Right: no box can be
    placed, so the algorithm terminates.}
  \label{fig:quasi-triangulated}
\end{figure}


\subsection{The \fbcrhier{} problem}
\label{sec:hierarchy}

The \fbcrhier{} problem is a restricted variant of the \fbcr{} problem
that can be used to create word clouds with a hierarchical structure;
see Fig.~\ref{fig:complexity-classes}.
The input is a directed acyclic graph $G$ with only one sink and with
a plane embedding.  The task is to find a representation that
\emph{hierarchically realizes $G$}, meaning that for each directed
edge $(v,u)$ in $G$ the top of the box for $v$ is in contact with the
bottom of the box for $u$.

If the embedding of $G$ is not fixed, the problem is \NP-hard
even for a tree, by an easy adaptation of the proof of
Theorem~\ref{thm:trees:hardness}.
(Remove the vertices $a_2, a_3, a_4$, and orient the remaining edges
of $T_S$ upward according to the representation shown in
Fig.~\ref{fig:tree:hardness}.)
However, if we fix the embedding of the supporting graph $G$,
then \fbcrhier{} can be solved efficiently.

\begin{theorem}\label{thm:hierarchical-planar}
  \fbcrhier{} can be solved in polynomial time.
\end{theorem}

\begin{proof}  
  Let $G$ be the given supporting graph,
  with vertices corresponding to boxes
  $B_1,\ldots,B_n$ where $B_i$ has height $h_i$ and width $w_i$, and
  $B_1$ is the unique sink. 
  We first check that the orientation and embedding of $G$ are
  compatible, that is, that incoming edges and outgoing edges are
  consecutive in the cyclic order around each vertex.

  The main idea is to set up a system of linear equations for the $x$-
  and $y$-coordinates of the sides of the boxes.  Let variables $t_i$
  and $b_i$ represent the $y$-coordinates of the top and bottom of
  $B_i$ respectively, and variables $\ell_i$ and $r_i$ represent the
  $x$-coordinates of the left and right of $B_i$ respectively.  For
  each $i=1,\dots,n$, impose the linear constraints
  $t_i = b_i + h_i$ and $r_i = \ell_i + w_i$.  For each directed edge
  $(B_i, B_j)$, impose the constraints
  $t_i=b_j, r_i > \ell_j$, and $r_j > \ell_i$.  The last two
  constraints force $B_i$ and $B_j$ to share some $x$-range in which
  they can make vertical contact.  Initialize $t_1=0$.

  With these equations, variables $t_i$ and $b_i$ are completely
  determined since every box $B_i$ has a directed path to $B_1$.
  Furthermore, the values for $t_i$ and $b_i$ can be found using a
  depth-first-search of $G$ starting from $B_1$.

  The $x$-coordinates are not yet determined and depend on the
  horizontal order of the boxes, which can be established as follows.
  We scan the boxes from top to bottom, keeping track of the
  left-to-right order of boxes intersected by a horizontal line that
  sweeps from $y=0$ downwards.  Initially the line is at $y=0$ and
  intersects only~$B_1$.  When the line reaches the bottom of a box
  $B$, we replace $B$ in the left-to-right order by all its
  predecessors in $G$, using the order given by the plane embedding.
  In case multiple boxes end at the same $y$-coordinate, we make the
  update for all of them.  Whenever boxes $B_a$ and $B_b$ appear
  consecutively in the left-to-right order, we impose the constraint
  $r_a \le \ell_b.$

  The scan can be performed in $O(n \log n)$ time using a priority
  queue to determine which boxes in the current left-to-right order
  have maximum $b_i$ value.  The resulting system of equations has
  size $O(n)$ (because the constraints correspond to edges of a planar
  graph).  It is straightforward to verify that the system of
  equations has a solution if and only if there is a representation of
  the boxes that hierarchically realizes $G$.  The constraints define
  a linear program (LP) and can be solved efficiently.  (A feasible
  solution can be found faster than with an LP, but we omit the details
  in this paper.)
\end{proof}


\rotate{We can show that \fbcrhier becomes weakly
  NP-complete if rectangles may be rotated, by a simple reduction from
  \prob{Subset Sum} (details in Appendix~\ref{sub:hier}).  }

\section{The \fbcropt{} problem}\label{sec:optimize}

In this section, we study approximation algorithms for
\fbcropt\extremal{ and consider an extremal variant of the problem}.

\subsection{Approximation Algorithms.}

We present approximation algorithms for \fbcropt
restricted to certain graph classes.  Our basic building
blocks are an approximation algorithm for stars and an exact
algorithm for cycles.  Our general technique is to find a collection
of disjoint stars or cycles in a graph.  We begin with stars, using a
reduction to
the \prob{Maximum Generalized Assignment Problem} (GAP) defined as
follows:
Given a set of bins with capacity constraints and a set of items that
may have different sizes and values in each bin, pack a
maximum-value subset of items into the bins. It is known that the
problem is \NP-hard (\prob{Knapsack} and \prob{Bin
  Packing} are special cases of \prob{GAP}), and there exists an
$(1-1/e)$-approximation
algorithm~\cite{Fleischer2011}. In the remainder, we assume that there
is an $\alpha$-approximation algorithm for \prob{GAP},
setting $\alpha = 1-1/e > 0.632$.

\begin{theorem}\label{thm:approx-star}
  There exists an 
  $\alpha$-approximation algorithm for \fbcropt{} on stars.
\end{theorem}
\begin{proof}
  Let $B_0$ denote the box corresponding to the center of the star. In
  any optimal solution for the \fbcropt{} problem there are four boxes
  $B_1,B_2,B_3,B_4$ whose sides contain one corner of $B_0$
  each. Given $B_1,B_2,B_3,B_4$, the problem reduces to assigning each
  remaining box $B_i$ to 
  one of the four sides of $B_0$, where it makes contact for its whole
  length;
  see Fig.~\ref{fig:approx-star}.

  \begin{figure}[t]
    \centering
    \includegraphics[width=8cm]{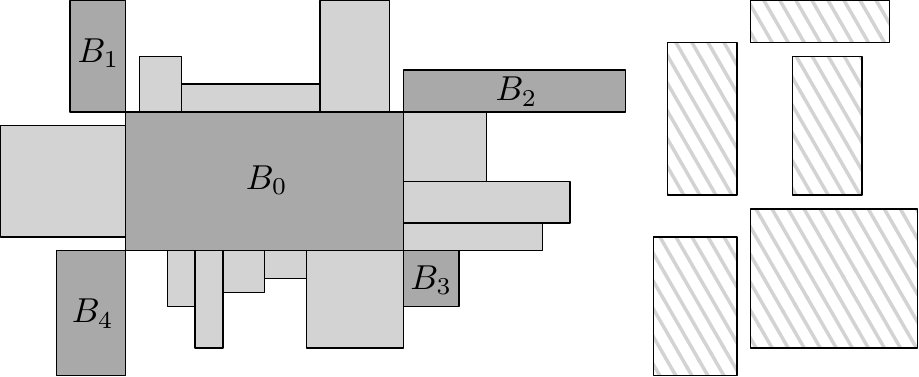}
    \caption{An optimal representation for the \fbcropt{} problem whose
      supporting graph is a star with center $B_0$. The striped boxes
      did not fit into the solution.}
    \label{fig:approx-star}
  \end{figure}

  This is a special case of \prob{GAP}: The bins are the
  four sides of $B_0$,
  the size of an item is its width for the horizontal bins and its
  height for the vertical bins,
  and the value of an item is the profit of its adjacency to the
  central box.
  We can now apply the algorithm for the \prob{GAP} problem, which
  gives an $\alpha$-approximation for the set of boxes. To get an
  approximation for the \fbcropt{} problem, we consider all possible
  ways of choosing boxes $B_1,B_2,B_3,B_4$, which increases the runtime
  only by a polynomial factor.
\end{proof}

\rotate{In the case where rectangles may be rotated by $90^\circ$,
  the \fbcropt{} problem on a star reduces to an easier problem, the
  \prob{Multiple Knapsack Problem}, where every item has the same size
  and value no matter which bin it is placed in.  This is because we
  will always attach a rectangle $B$ to the central rectangle of the
  star using the smaller dimension of $B$.  There is a PTAS for
  \prob{Multiple Knapsack}~\cite{Chekuri}.  Therefore, there is a PTAS
  for \fbcropt on stars if we may rotate rectangles.}

A \emph{star forest} is a disjoint union of stars.
Theorem~\ref{thm:approx-star} applies to a star forest since we can
combine the solutions for the disjoint stars.

\begin{theorem}
  \label{thm:approx-from-stars}
  \fbcropt on the class of graphs that can be partitioned in
  polynomial time into $k$ star forests admits an
  $\alpha/k$-approximation algorithm.
\end{theorem}
\begin{proof}
  The algorithm is to partition the edges of the supporting graph into
  $k$ star forests, apply the approximation algorithm of
  Theorem~\ref{thm:approx-star} to each star forest, and take the best
  of the $k$ solutions.  This takes polynomial time.  We claim this
  gives the desired approximation factor.  Consider an optimum
  solution, and let \Wopt be the total profit of edges that are
  realized as contacts.  By the pigeon hole principle, there is a star
  forest $F$ in the partition with realized profit at least
  $\Wopt/k$ in the optimum solution.  Therefore our approximation
  achieves at least $\alpha \Wopt/k$ profit for~$F$.
\end{proof}

\begin{corollary}
  \label{cor:approx}
  \fbcropt admits
  \begin{itemize}
  \item an $\alpha/2$-approximation algorithm on trees,
  \item an $\alpha/5$-approximation algorithm on planar graphs.
  \end{itemize}
\end{corollary}
\begin{proof}
  It is easy to partition any tree into two star forests in linear
  time. Moreover, it is known that every planar graph has star
  arboricity at most $5$, that is, it can be partitioned into at most
  $5$ star forests, and such a partition can be found in polynomial
  time~\cite{Hakimi199693}.  The results now follow directly from
  Theorem~\ref{thm:approx-from-stars}.
\end{proof}

Our star forest partition method is possibly not optimal. Nguyen
\textit{et al.}~\cite{nguyen2008approximating} show how to find a star
forest of an arbitrary weighted graph carrying at least half of the
profits of an optimal star forest in polynomial-time.  We can't, however,
guarantee that the approximation of the optimal star forest
carries a positive fraction of the total profit in an optimal solution
to \fbcropt.  Hence, approximating \fbcropt
for general graphs remains an open problem. As a first step into this
direction, we present a constant-factor approximation for supporting
graphs with bounded maximum degree. First we need the following lemma.

\begin{lemma}\label{lem:build-cycle}
  Given a sequence of $n \geq 3$ boxes, we can find a representation
  realizing the $n$-cycle in linear time.
\end{lemma}
\begin{proof}
  Let $C = (v_1, v_2, \ldots, v_n)$ be a cycle.  Let $W$ be the sum of
  all the widths, $W=\sum_i w_i$, and let $t$ be maximum index such
  that $\sum_{i \le t} w_i < W/2$.  We place $v_1, v_2, \ldots, v_t$
  side by side in order from left to right with their bottoms on a
  horizontal line $h$.  We call this the ``top channel''.  Starting
  from the same point on $h$ we place $v_n, v_{n-1}, \ldots, v_{t+2}$
  side by side in order from left to right with their tops on $h$.  We
  call this the ``bottom channel''.  Note that $v_1$ and $v_n$ are in
  contact.  It remains to place $v_{t+1}$ in contact with $v_t$ and
  $v_{t+2}$. It is easy to show that the following works: add
  $v_{t+1}$ to the channel of minimum width, or in case of a tie,
  place $v_t$ straddling the line $h$.
\end{proof}

\begin{figure}[t]
  \centering
  \includegraphics{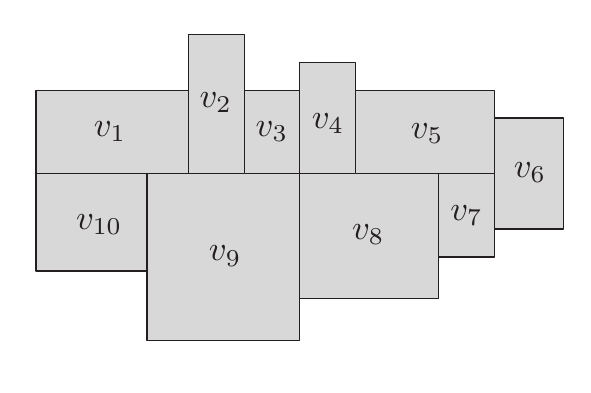}\hspace{1cm}
  \includegraphics[width=5cm]{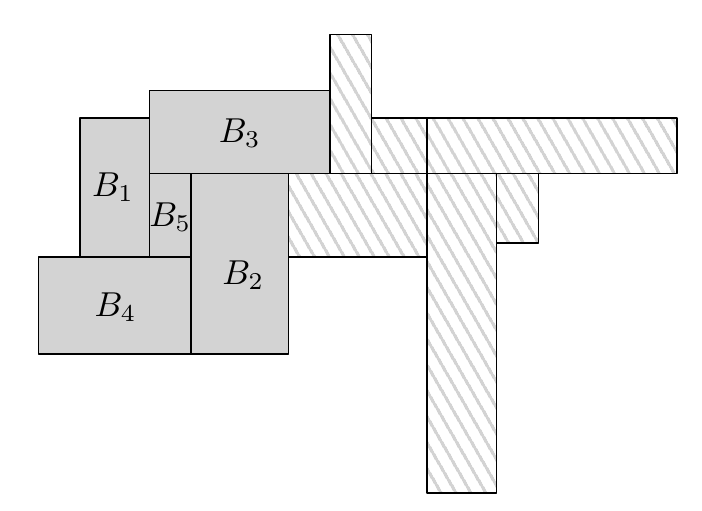}
  \caption{Left: Realizing cycle $(v_1,\ldots,v_{10})$. Right: $8$
    adjacencies with $5$ boxes in Theorem~\ref{thm:extremal}.}
  \label{fig:combined}
\end{figure}

Following the idea of Theorem~\ref{thm:approx-from-stars}, we can
approximate \fbcropt{} by applying Lemma~\ref{lem:build-cycle} to a
partition of the supporting graph into sets of disjoint cycles.

\begin{theorem}\label{thm:approx-from-cycles}
  \fbcropt on the class of graphs that can be partitioned into $k$
  sets of disjoint cycles (in polynomial time) admits a
  (polynomial-time) algorithm that achieves total profit at least
  $\frac{1}{k} \sum_{i\neq j} p_{ij}$.  In particular, there is
  a 
  $1/k$-approximation algorithm for \fbcropt on this graph class.
\end{theorem}

\begin{corollary}
  \label{cor:delta-approx}
  \fbcropt{} on graph of maximum degree~$\Delta$ admits a
  $2/(\Delta+1)$-approximation.
\end{corollary}
\begin{proof}
  As Peterson~\cite{peterson} shows, the edges of any graph of maximum
  degree $\Delta$ can be covered by $\lceil\Delta/2\rceil$
  sets of cycles, and such sets can be found in polynomial time. The
  result now follows from Theorem~\ref{thm:approx-from-cycles}.
\end{proof}

\extremal{
  \subsection{An Extremal \fbcropt{} Problem.}

  In the following, we bound the maximum number of contacts that can
  be made when placing $n$ boxes.
  It is easy to see that for $n=2,3$ any set of boxes allows $2n-3$
  contacts. In case $n=4$ the boxes can be arranged so that their
  corners meet at a point, thus realizing $2n-2$ contacts. For larger
  $n$ we have:

\begin{theorem}\label{thm:extremal}
  For $n\ge 5$ and any set of $n$ boxes, the boxes can be placed in
  the plane to realize $2n-2$ contacts.  For some sets of boxes this
  is the best possible.
\end{theorem}
\begin{proof}
  Let $B_1, \ldots, B_n$ be any set of boxes.  We place the first 5
  boxes to make 8 contacts, and place the remaining boxes to make 2
  contacts each for a total of $8+2(n-5) = 2n-2$ contacts.  Among the
  first 5 boxes, let $B_1$ and $B_2$ be the boxes with largest height,
  and $B_3$ and $B_4$ be the boxes with largest width. Place the five
  boxes as in Fig.~\ref{fig:combined}.  Place the remaining boxes one
  by one as in the proof of Lemma~\ref{lem:build-cycle} along the
  horizontal line between $B_2$ and $B_3$.  Then each remaining box
  makes two new contacts.

  Next we describe a set of $n$ boxes for which the maximum number of
  contacts is $2n-2$. Let $B_i$ be a square box of side length
  $2^i$. Consider any placement of the boxes and partition the
  contacts into horizontal and vertical contacts. Here we assume that
  a point contact of two boxes is horizontal if the point is the
  south-west corner of the first box and the north-east corner of the
  second; otherwise, a point contact is vertical. From the side
  lengths of boxes, it follows that neither set of contacts contains a
  cycle.  Thus each set of contacts has size at most $n-1$ for a total
  of $2n-2$.
\end{proof}
}

\section{The \fbcrarea{} problem}\label{sec:area}
The same supporting graph can often be realized by different contact
representations, not all of which are equally useful or visually
appealing when viewed as word clouds.  In this section we consider the
\fbcrarea{} problem and show that finding a ``compact'' representation
that fits into a small bounding box is another \NP{}-hard problem.

The reduction is from the (strongly) \NP{}-hard $2$D \prob{Strip
  Packing} problem~\cite{LMM02}: The input is a set $R$
of $n$ rectangles with height and weight functions $w: R \rightarrow
\N$ and $h: R \rightarrow \N$, and a strip of width $W$ and height
$H$.  All the input numbers are bounded by some polynomial in $n$.
The task is to pack the given rectangles into the
strip. 

The \prob{Strip Packing} problem is actually equivalent to \fbcrarea{}
when the supporting graph is an independent set.  However, edges in
the supporting graph impose additional constraints on the
representation, which might make \fbcrarea{} easier. The following
theorem (proved in the appendix) shows that this is not the case.

\wormhole{path-packing-hard}
\begin{theorem}
  \label{thm:path-packing-hard}
  \fbcrarea is \NP-hard even on paths.
\end{theorem}

\section{Experimental Results}\label{sec:experimental}

We implemented our new methods for constructing word clouds: the
\textsc{Star Forest} algorithm based on extracting star forests
(Corollary~\ref{cor:approx}), and the \textsc{Cycle Cover} algorithm
based on decomposing edges of a graph into cycle covers
(Theorem~\ref{thm:approx-from-cycles}).  We compared the algorithms
with the existing method from~\cite{wordle09} (referred to as
\textsc{Random}), the algorithm from~\cite{Cui_2010_wordcloud}
(referred to as \textsc{CPDWCV}), and the algorithm
from~\cite{wu2011semantic} (referred to as \textsc{Seam Carving}).
Our dataset is 120 Wikipedia documents, with 400 words or more. For
the word clouds, we removed stop-words (e.g., ``and'', ``the'',
``of''), and constructed supporting graphs $G_{50}$ and $G_{100}$ for
$50$ and $100$ the most frequent words respectively. Implementation
details are provided in the appendix.

We compare the percentage of realized profit in the representation of
the supporting graphs.  Since \textsc{Star Forest} handles planar
supporting graphs, we first extract a maximal planar subgraph of $G$,
and then apply the algorithm on the subgraph. The percentage of
realized profit is presented in the table.  Our results indicate that,
in terms of the realized profit, \textsc{Cycle Cover} and \textsc{Star
  Forest} outperform existing approaches; see Fig.~\ref{fig:wordle}.
In practice, \textsc{Cycle Cover} realizes more than $17 \%$ of the
total profit of graphs with $50$ vertices. On the other hand, existing
algorithms may perform better in terms of compactness, aspect ratio,
and other aesthetic criteria; we leave a deeper comparison of word
cloud algorithms as a future research direction.

\smallskip

\begin{center}
  \begin{tabular}{@{}lr<{\qquad}r<{\quad\qquad}@{}}
    Algorithm & \multicolumn{1}{c}{Realized Profit of $G_{50}$} &
    \multicolumn{1}{c}{~~~~~Realized Profit of $G_{100}$} \\
    \toprule
    \textsc{Random}~\cite{wordle09} &       $3.4 \%$ & $2.2 \%$ \\
    \textsc{CPDWCV}~\cite{Cui_2010_wordcloud} &       $12.2 \%$ & $8.9 \%$ \\
    \textsc{Seam Carving}~\cite{wu2011semantic} & $7.4 \%$ & $5.2 \%$ \\
    \textsc{Star Forest} &  $11.4 \%$ & $8.2 \%$ \\
    \textsc{Cycle Cover} &  $17.8 \%$ & $13.8 \%$ \\
  \end{tabular}
\end{center}

\section{Conclusions and Future Work}
\label{sec:conclusions}


We formulated the Word Rectangle Adjacency Contact (\fbcr{}) problem,
motivated by the desire to provide theoretical guarantees for
semantics-preserving word cloud visualization. We described efficient
algorithms for variants of \fbcr, showed that some
variants are \NP-hard, and presented several approximation
algorithms.  A natual open problem is to find an approximation
algorithm for general graphs with arbitrary profits.

\medskip\noindent{\bf Acknowledgments.} Work on this problem began at
Dagstuhl Seminar 12261. We thank the organizers, participants, Therese
Biedl, Steve Chaplick, and G\"unter Rote.
%

\bibliographystyle{alpha}
\bibliography{literature,refs}

\appendix

\newpage\noindent{\LARGE\bf Appendix}



\rotate{
  \section{The \fbcr{} problem}
  \label{sec:fbcr}

  \smallskip\noindent Theorem~\ref{thm:trees:hardness} still holds in
  the case where rectangles may be rotated.  The construction uses
  squares for the $a_i$'s.  Rectangle $c$ has height 1/2 and all the
  other rectangles have both width and height greater than 1/2 so no
  rectangle can make contact along the sides of $c$ in-between the
  $a_k$'s.  Finally, all rectangles have height greater than width, so
  there is no advantage to rotating any rectangle.}

\subsection{The \fbcr{} problem on irreducible triangulations}
\label{sub:irreducible}

\begin{backInTime}{quasi-triangulated}
\begin{theorem}
  \fbcr on irreducible triangulations can be solved in linear time.
\end{theorem}
\end{backInTime}

\begin{proof} 
  Let $G$ be the supporting graph, an irreducible triangulation. We
  consider $G$ embedded in the plane with outer face
  $\{v_N,v_E,v_S,v_W\}$. Note that this embedding is unique. By abusing
  notation, we refer to a vertex and its corresponding box with the
  same letter.

  We begin by placing a horizontal and a vertical ray emerging from
  the same point in positive $x$-direction and positive $y$-direction,
  respectively. For the first phase of the algorithm let us pretend
  that the horizontal ray is the box $v_S$ (imagine a rectangle with
  tiny height and huge width) and the vertical ray is the box $v_W$
  (imagine a rectangle with tiny width and huge height), independent
  of how the actual boxes look like; see
  Fig.~\ref{fig:quasi-triangulated}.

  We build up a representation by adding one rectangle at a time. At
  every intermediate step the representation is \emph{rectilinear
    convex}, that is, its intersection with any horizontal or vertical
  line is connected. In other words, the representation has no holes
  and a ``staircase shape''. We maintain the set of all
  \emph{concavities}, that is, points on the boundary of the
  representation, which are bottom-right or top-left corners of some
  rectangle but not a top-right corner of any rectangle. Initially
  there is only one concavity, namely the point where the rays $v_W$
  and $v_S$ meet.

  Each concavity $p$ is a point on the boundary of two rectangles, say
  $u$ and $v$. Since $G$ has no separating triangles there are exactly
  two vertices that are adjacent to both, $u$ and $v$, or only one if
  $\{u,v\} = \{v_S,v_W\}$. For exactly one of the these vertices, call
  it $w$, the rectangle is not yet placed because its bottom-left
  corner is supposed to be placed on the concavity $p$. We say that
  $w$ \emph{fits into the concavity $p$}. We call a vertex $w$
  \emph{applicable} to an intermediate representation if it fits into
  some concavity and adding the rectangle $w$ gives a representation
  that is rectilinear convex. In the very beginning the unique common
  neighbor of $v_S$ and $v_W$ is applicable.

  The algorithm proceeds in $n-4$ steps as follows. At each step we
  identify a inner vertex $w$ of $G$ that is applicable to the current
  representation. We add the rectangle $w$ to the representation and
  update the set of concavities and applicable vertices. At most two
  points have to be added to the set of concavities, while one is
  removed from this set. The vertices that fit into the new
  concavities can easily be read off from the plane embedding of
  $G$. Checking whether these vertices are applicable is easy. If the
  top-left or bottom-right corner of $w$ does not define a concavity
  then one has to check whether the vertices that fit into existing
  concavities to the left or below, respectively, are now
  applicable. So each step can be done in constant time.

  If the algorithm has placed the last inner vertex, it suffices to
  check whether the representation without the two rays is a
  rectangle, that is, whether there are exactly two concavities
  left. If so, call this rectangle $R$, we check whether the width of
  $R$ is at most the width of $v_N$ and $v_S$ and whether the height
  of $R$ is at most the height of $v_E$ and $v_W$. If this holds true,
  we can easily place the rectangles $v_N$, $v_E$, $v_S$, $v_W$ to get
  a representation that realizes $G$. The total running time is
  linear.

  On the other hand, if the algorithm stops because there is no
  applicable vertex, or the height/width-conditions in the end phase
  are not met, then there is no representation that realizes $G$. This
  is due to the lack of choice in building the representation -- if a
  vertex $v$ is applicable to a concavity $p$ then the bottom-left
  corner of $v$ has to be placed at $p$ in order to establish the
  contacts of $v$ with the two rectangles containing $p$.
\end{proof}


\rotate{
  \subsection{The \fbcrhier problem}
  \label{sub:hier}

  \smallskip\noindent We justify the claim that \fbcrhier
  becomes weakly \NP{}-complete if rectangles may be rotated.
  We use a reduction from \prob{Subset Sum}.  Given an instance of
  \prob{Subset Sum} consisting of $n$ items $s_i, i=1, \ldots , n$ and
  a desired sum $S$, construct a large top square $T$ and a large
  bottom square $B$, and a square $M$ of side-length $n+S$ that must
  lie between $B$ and $T$.  Add a chain of rectangles $B_1, \ldots,
  B_n$ from $B$ to $T$ where $B_i$ has dimensions $1 \times (1+s_i)$.
  Rectangles $B_i$ that are oriented to have height $1+s_i$ correspond
  to ``chosen'' elements for the \prob{Subset Sum} problem.  Via this
  correspondence, the \prob{Subset Sum} instance has a solution if and
  only if the constructed \fbcrhier{} instance has a solution.  }

\section{The \fbcrarea{} problem}

\begin{backInTime}{path-packing-hard}
\begin{theorem}
  \fbcrarea is \NP-hard even on paths.
\end{theorem}
\end{backInTime}

\begin{proof} 
  We use a reduction from \prob{Strip Packing}, so fix any instance
  $I$ of \prob{Strip Packing} consisting of rectangles
  $r_1,\ldots,r_n$ and two integers $H$ and $W$. Let $d =
  {\epsilon}/{\max(W,H)}$ for some $\epsilon \in (0,1)$.

  We define an instance of the \fbcrarea{} problem by slightly
  increasing the heights and widths in $I$. The idea is to lay a unit
  square grid over the strip and blow each grid line up to have a
  thickness of $d$; see Fig.~\ref{fig:paths-hard-grid}. Each rectangle
  in $I$ is stretched according to the number of grid lines is
  intersects.

  \begin{figure}[b]
    \centering
    \subfloat{\includegraphics[width=0.4\textwidth]{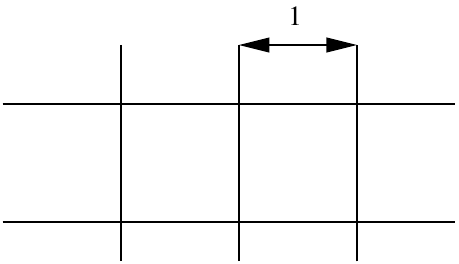}}
    \hfil
    \subfloat{\includegraphics[width=0.4\textwidth]{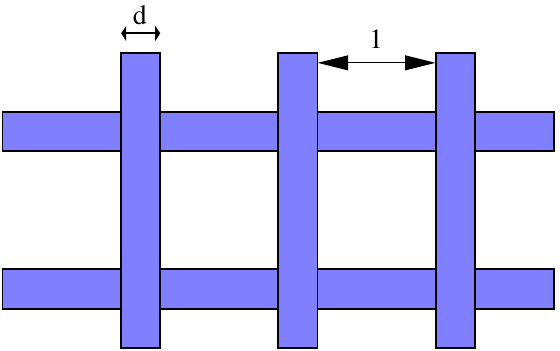}}
    \caption{Grid before and after stretching}
    \label{fig:paths-hard-grid}
  \end{figure}

  More precisely, we define for $i=1,\ldots,n$ a rectangle $r'_i$ of
  width $w(r_i) + (w(r_i) - 1) d$ and height $h(r_i) + (h(r_i) - 1)
  d$. Further we define $W' = W + (W-1)d$ and $H' = H +
  (H-1)d$. Finally, we arrange the rectangles $r'_1,\ldots,r'_n$ into a
  path $P$ by introducing between $r_i$ and $r_{i+1}$
  ($i=1,\ldots,n-1$), as well as before $r'_1$ $k$ small $x \times x$
  square, called \emph{connector squares}. We choose $k$ and $x$ to
  satisfy

  \begin{align}
    kx &= 4(n+3)(H + 2nW) \hspace{4em}\text{and}\label{eqn:connector-length}\\
    n(kx^2 + 2x) &= d.  \label{eqn:connector-space}
  \end{align}

  In particular, we choose
  \begin{align*}
    x &= \frac{d}{2n(2Hn+6H+4n^2W+12nW+1)} \hspace{4em}\text{and}\\
    k &= \frac{4(n+3)(H+2nW)}{x}.
  \end{align*}

  We claim that there is a representation realizing $P$ within the $W'
  \times H'$ bounding box if and only if the original rectangles
  $r_1,\ldots,r_n$ can be packed into the original $W \times H$
  bounding box.

  First consider any representation realizing $P$ within the $W' \times
  H'$ bounding box and remove all connector squares from it. Since $W'
  < W + \epsilon < W + 1$ and $H' < H + \epsilon < H + 1$, the
  stretched bounding box has the same number of grid lines than the
  original. Hence the rectangles $r'_1,\ldots,r'_n$ can be replaced by
  the corresponding rectangles $r_1,\ldots,r_n$ and perturbed slightly
  such that every corner lies on a grid point. This way we obtain a
  solution for the original instance of \prob{Strip Packing}.

  Now consider any solution for the \prob{Strip Packing} instance, that
  is, any packing of the rectangles $r_1,\ldots,r_n$ within the $W
  \times H$ bounding box. We will construct a representation realizing
  the path $P$ within the $W' \times H'$ bounding box. We start blowing
  up the grid lines of the $W \times H$ bounding box to thickness $d$
  each, which also effects all rectangles intersected by a grid line in
  its interior. This way we obtain a placement of bigger rectangles
  $r'_1,\ldots,r'_n$ $I'$ in the bigger $W' \times H'$ bounding box,
  such that every rectangle $r'_i$ intersects the interiors of exactly
  those blown-up grid lines corresponding to the grid lines that
  intersect $r_i$ interiorly. Thus any two rectangles $r'_i$ and $r'_j$
  are separated by a vertical or horizontal corridor of thickness at
  least $d$. We will refer to the grid lines of thickness $d$ as
  \emph{gaps}.

  It remains to place all the connector square so as to realize the
  path $P$. The idea is the following. We start in the lower left
  corner of the bounding box, and lay out connector squares
  horizontally to the right inside the bottommost horizontal gap until
  we reach the vertical gap that contains the lower-left corner of
  $r'_1$. We then start laying out the connector squares inside this
  vertical gap upwards, until we reach the lower-left corner of
  $r'_1$. Whenever a rectangle $r'_i$ overlaps with this vertical gap,
  we go around $r'_i$; see
  Fig.~\ref{fig:paths-conn-rerouting-after}. This way we lay out at
  most $(3W'+H')/x$ connector squares, which
  by~\eqref{eqn:connector-length} is less than $k$. The remaining
  connector squares are ``folded up'' inside the vertical gap; see
  Fig.~\ref{fig:hardness-path-folding}.

  \begin{figure}[tb]
    \centering
    \subfloat[\NP-hardness of \fbcrarea{} for
    paths]{\label{fig:hardness-path-packing}\includegraphics[width=.47\textwidth]{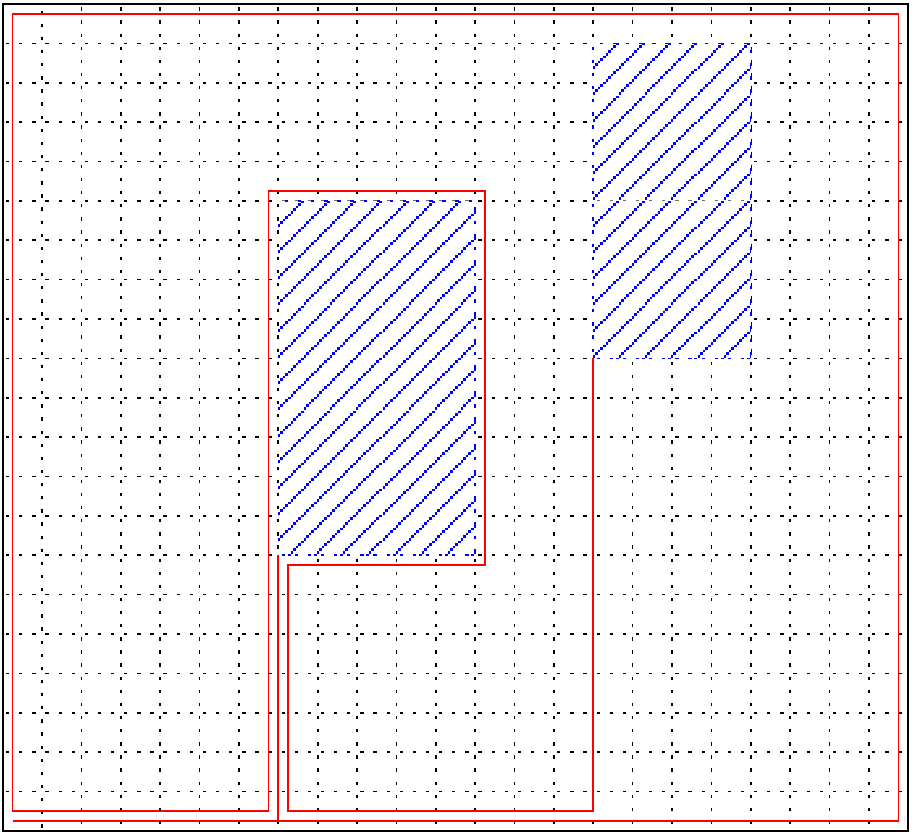}}

    \subfloat[folding connector rectangles inside a gap]{\label{fig:hardness-path-folding}\parbox[b]{.235\textwidth}{\centering\includegraphics[height=.3\textheight]{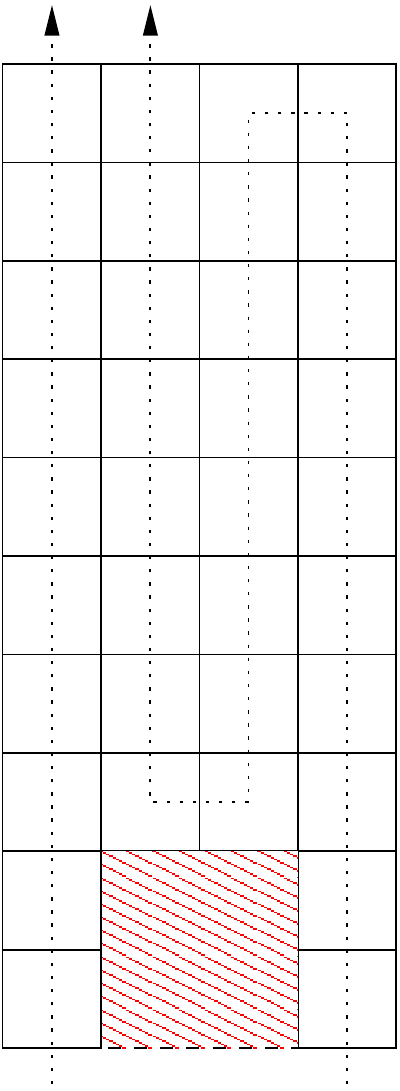}}}
    \hfill
    \subfloat[connectors before
    rerouting]{\label{fig:paths-conn-rerouting-before}\includegraphics[height=.3\textheight]{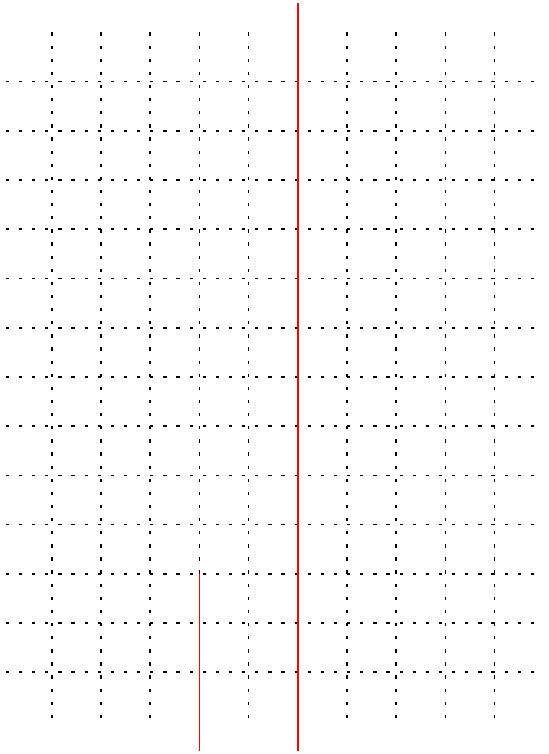}}
    \hfill \subfloat[connectors after
    rerouting]{\label{fig:paths-conn-rerouting-after}\includegraphics[height=.3\textheight]{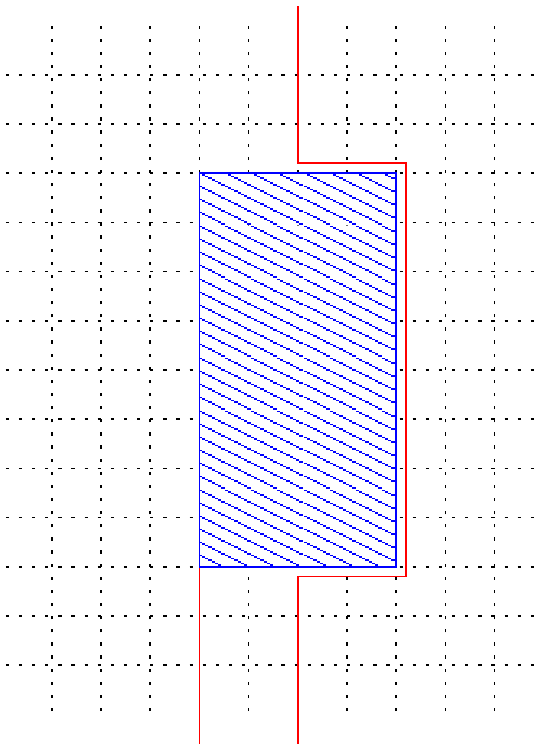}}

    \caption{Illustrations for Theorem~\ref{thm:path-packing-hard}.}
  \end{figure}

  Next we lay out the connectors squares between $r'_1$ and $r'_2$. We
  start where we ended before, that is, at the lower-left corner of
  $r'_1$, and go the along the path we took before till we reach the
  bottommost gap. Then we lay connector squares along the outermost
  gaps in counterclockwise direction, that is, first horizontally to
  the rightmost gap, then up to the topmost gap, left to the leftmost
  gap, and down to the bottommost gap. Now we do the same for $r'_2$
  than what we did for $r'_1$. If while going right we ``hit'' the
  connector squares going up to $r'_1$, we follow them up, go around
  $r'_1$, and go down again. This is possible since there are gaps all
  around $r'_1$; see Fig.~\ref{fig:hardness-path-packing}. Note that
  the red line of connectors will actually sit \textit{on} the dashed,
  expanded grid lines but are drawn next to them for better
  readability.

  We repeat the process for all the rectangles.

  We have to show two things: The number of connector squares between
  two $r'_i$ and $r'_{i+1}$ is large enough so that the length of the
  string of connectors is sufficient. And that the gaps have sufficient
  space so that we can fold up the connectors in them.

  The first condition is taken care of by
  equation~\eqref{eqn:connector-length}. We divide the path of the
  connectors in up to $n + 3$ parts: The first part
  $p_{{\text{down}}_i}$ is going down from $r'_i$ to the bottom
  gap. The second part $p_{\text{circle}}$ that goes around the
  bounding box in counterclockwise order to the vertical gap containing
  the lower-left corner of $r'_{i+1}$. This part is intercepted by up
  to $n$ parts $p_{{\text{avoid}}_k}$ where we hit a string of
  connectors going up to another rectangle $r'_k$ and we have to follow
  it, go around $r'_k$ and come down again. The last part
  $p_{{\text{up}}_{i+1}}$ is going up from the bottom gap to the
  position of $r'_{i+1}$. We will now show that each of these parts has
  a maximum length of $4(H' + 2nW')$.

  The parts $p_{{\text{up}}_{i+1}}$ and $p_{{\text{down}}_i}$ have to
  span the height $H'$ at most once, and may encounter all other
  rectangles $r'_k$ at most once. Going around any such $r'_k$ means at
  most traversing its width twice, which is at most $2W'$. Hence each
  of $p_{{\text{up}}_{i+1}}$ and $p_{{\text{down}}_i}$ has a total
  length of at most $H' + 2nW' < 4(H' + 2nW')$. Since every
  $p_{{\text{avoid}}_k}$ exactly follows the $p_{{\text{up}}_k}$, then
  surrounds $r'_k$ (which has maximum width $W'$ and maximum height
  $H'$) and then follows $p_{{\text{down}}_k}$, it has a maximum length
  of $2(W'+H') + 2(H' + 2nW') \leq 4(H' + 2nW')$. Finally,
  $p_{\text{circle}}$ has a maximum length of $2H' + 2W' \leq 4(H' +
  2nW')$.

  Thus, the total length of the path of connectors comprised of $n + 3$
  parts of at most length $4(H' + 2nW')$ each is at most $4(n+3)(H' +
  2nW')$. Equation~\eqref{eqn:connector-length} ensures that our string
  of connectors has sufficient length.

  The second condition is covered by
  equation~\eqref{eqn:connector-space}. Consider
  Fig.~\ref{fig:hardness-path-folding}. If a string of connectors just
  passes through a gap, it takes up exactly $1 \times x$ space. If it
  folds $m$ connector rectangles inside the gap, it takes $m \times
  x^2$ plus the 'wasted' space (the red shaded space in
  Fig.~\ref{fig:hardness-path-folding}). The wasted space can be at
  most $1 \times 2x$, and since every string of connectors has $k$
  connector rectangles, the space taken up by those can be at most
  $kx^2$, thus every string of connectors can take at most $kx^2 + 2x$
  space in any given gap. Since there are $n$ such strings of
  connectors and every gap has dimensions $1 \times d$,
  equation~\eqref{eqn:connector-space} ensures that the space in every
  gap is sufficient.

  We showed that we can find a layout of the path that corresponds to
  the optimum packing of the rectangles, if such a packing exists
  within the desired bounding box. Thus, finding the most
  space-efficient layout for a path of rectangles is \NP-hard.
\end{proof}

\section{Experimental Results}

Here we provide some details regarding the
implementation of the algorithms.

Before the algorithms are applied, the text is preprocessed using this
workflow: The text is split into sentences, and the sentences are
split into words using Apache OpenNLP. We then remove stop words,
perform stemming on the words and group the words with the same
stem. The similarity of words is computed using Latent Semantic
Analysis based on the co-occurrence of the words within the same
sentence.

In the implementation of \textsc{Star Forest}, we use the
$(\frac{\beta}{\beta + 1} - \epsilon)$-approximation
of Fleischer et al.~\cite{Fleischer2011} combined with a FPTAS for
\textsc{Knapsack} 
to approximate the stars.  In the implementation of \textsc{CPDWCV},
we achieved the best results in our experiments with parameters $K_r =
1000$ and $K_a = 25$. Some results are given in Fig.~\ref{fig:wordle}.

\begin{figure}[tb]
  \centering
  \subfloat[Random~Layout]{\fbox{\includegraphics[scale=.135]{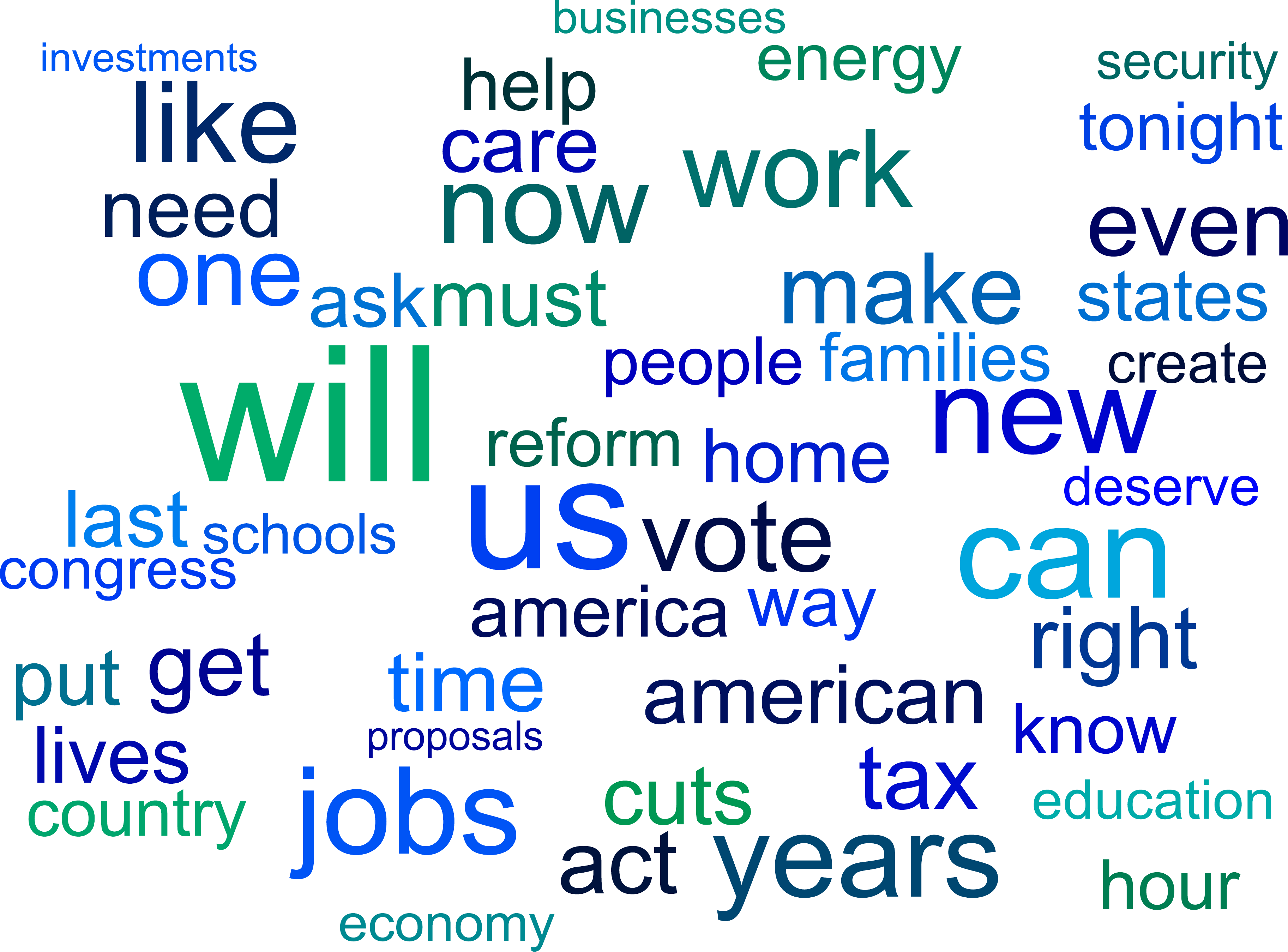}}}
  \hfill
  \subfloat[Seam~Carving~(\cite{wu2011semantic})]{%
    \fbox{\includegraphics[scale=.175]{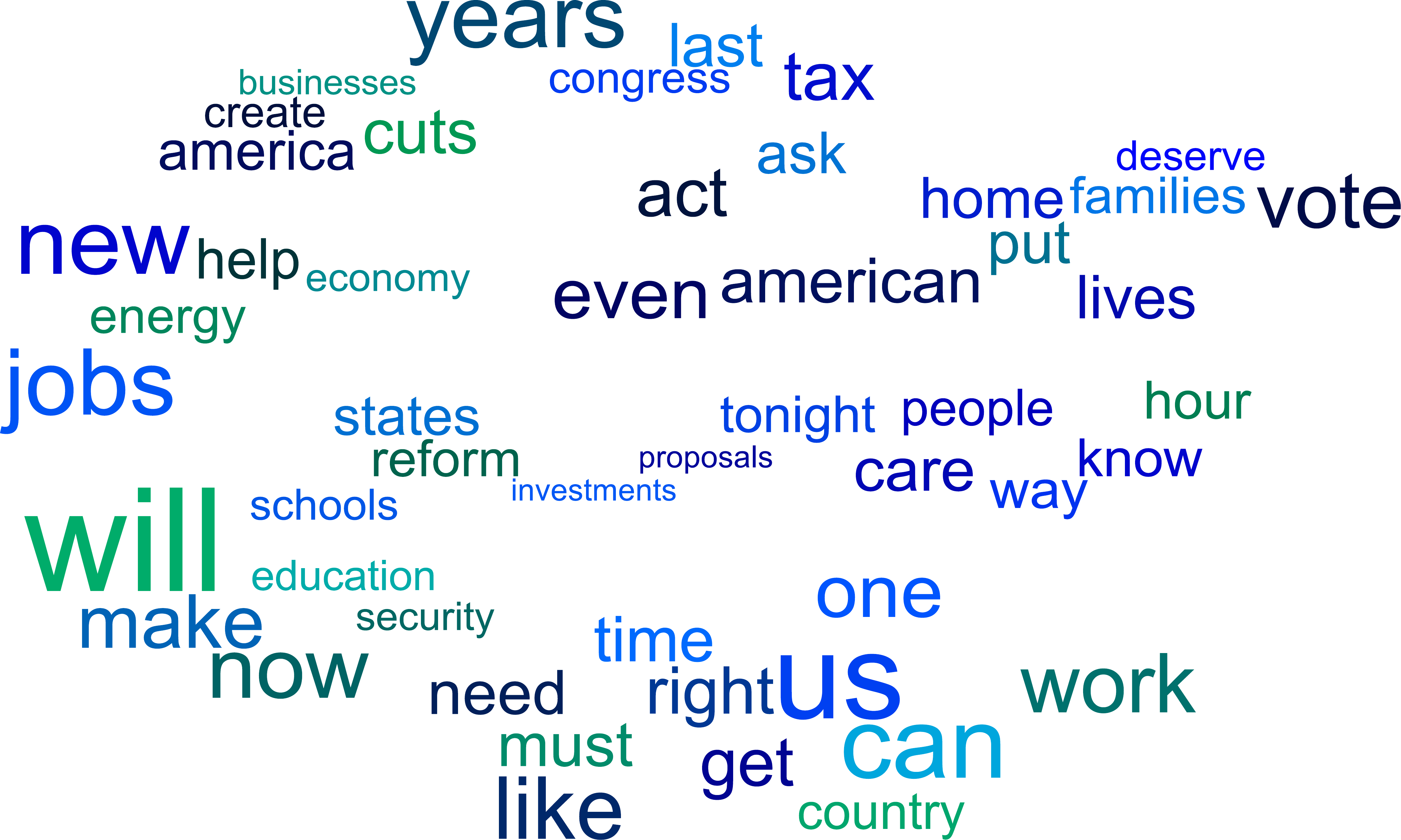}}}\\

  \bigskip

  \subfloat[CPDWCV (\cite{Cui_2010_wordcloud})]{%
    \fbox{\includegraphics[scale=.175]{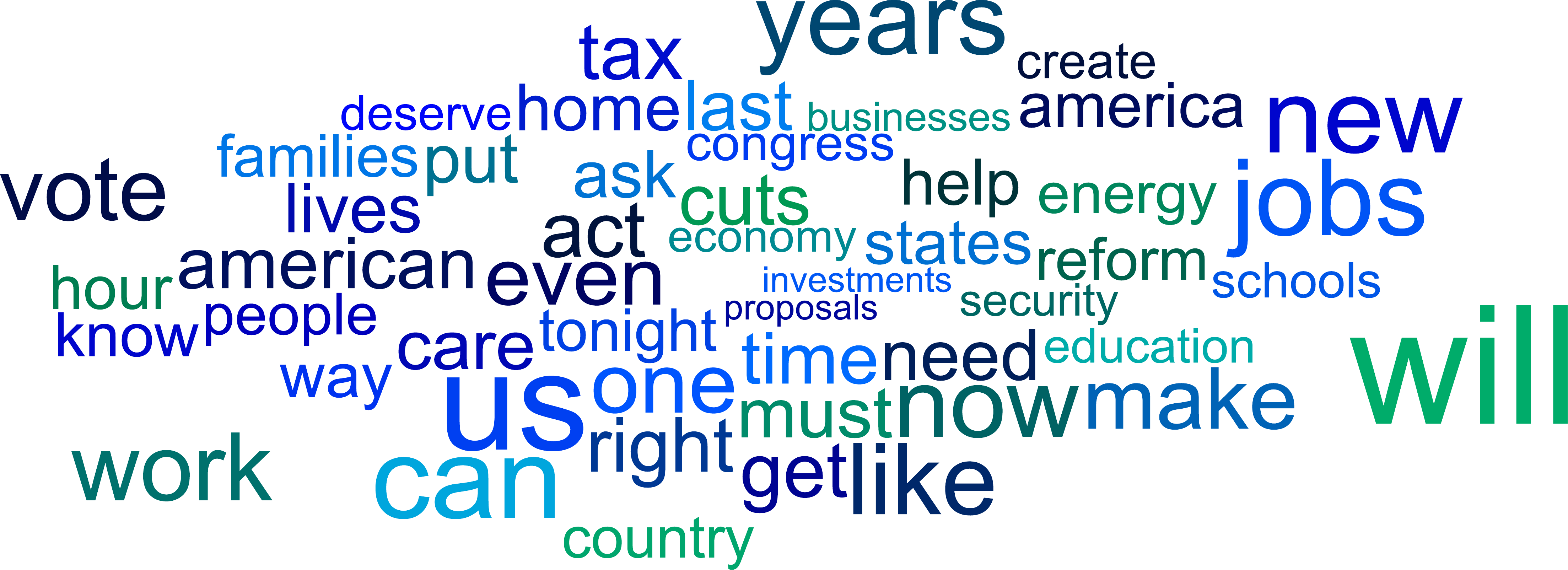}}} \hfill
  \subfloat[Star~Forest~(Corollary~\ref{cor:approx})]{%
    \fbox{\includegraphics[scale=.175]{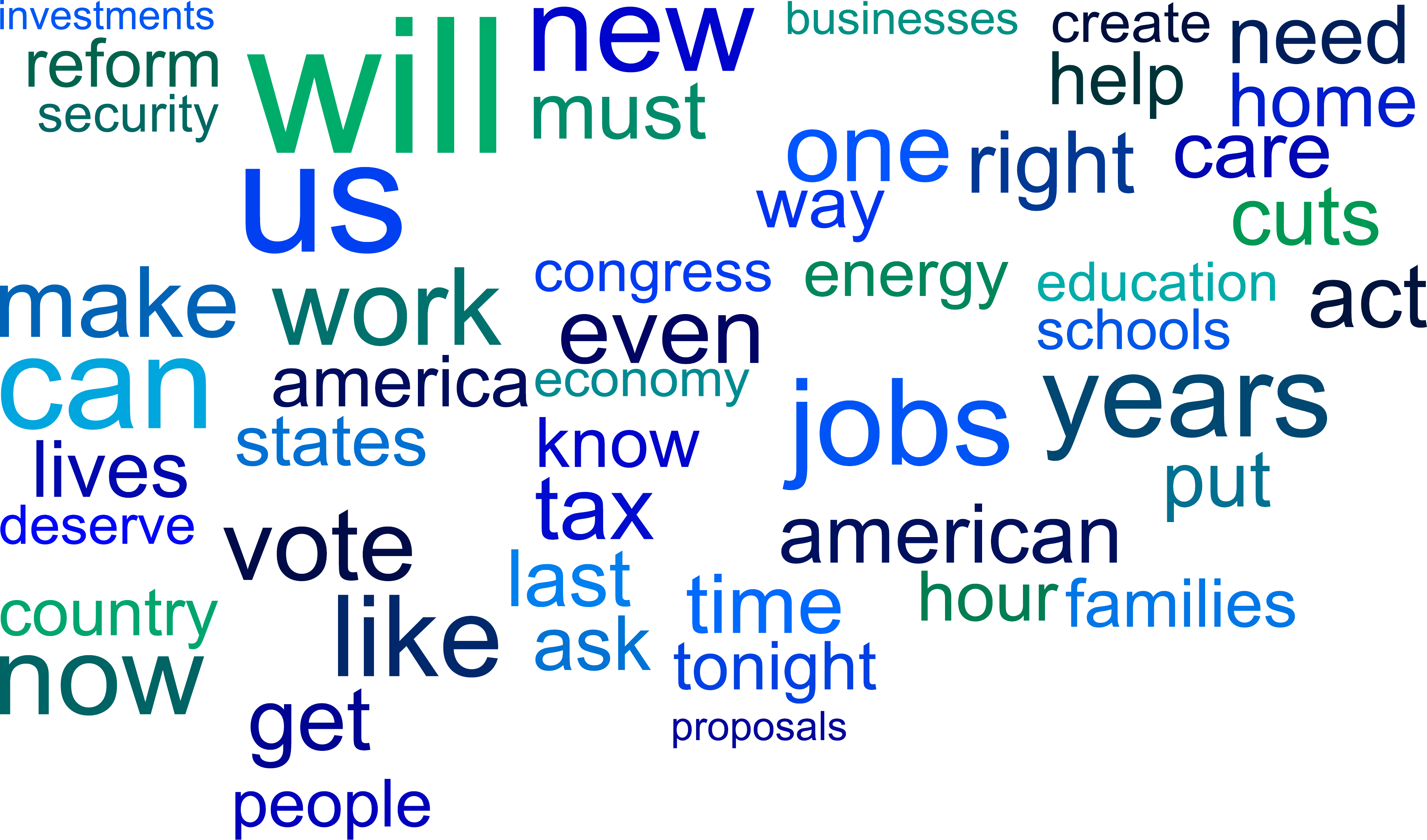}}}\\

  \bigskip

  \subfloat[Cycle~Cover~(Corollary~\ref{cor:delta-approx})]{%
    \fbox{\includegraphics[scale=.175]{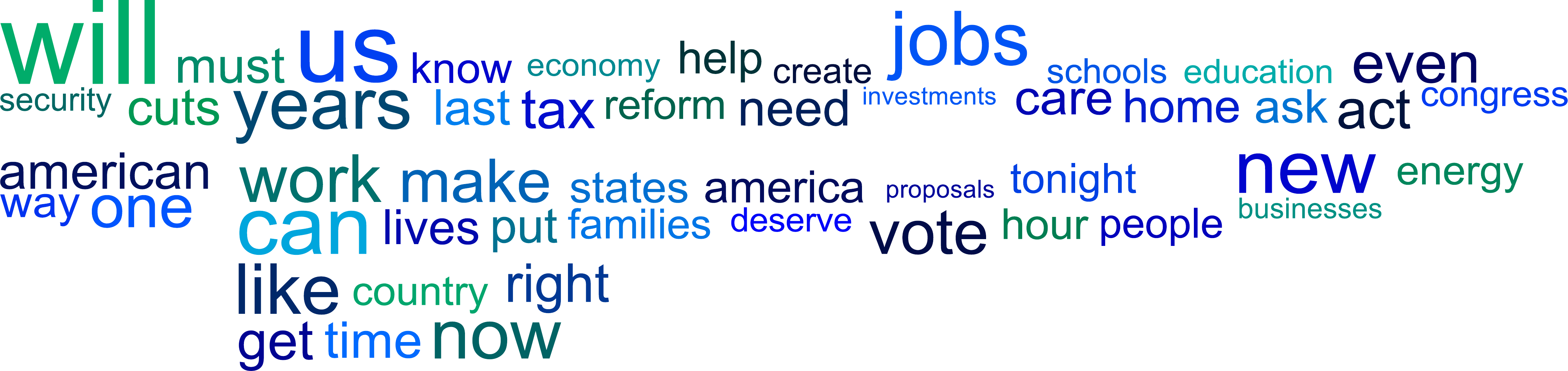}}}\\

  \bigskip

  \caption{Word clouds generated for Obama's 2013 State of the Union
    Speech by various algorithms. Percentage of the total realized
    profit: (a)~$2.3\%$ (b)~$7.1\%$ (c)~$11.8\%$ (d)~$12.7\%$
    (e)~$18.2\%$}
  \label{fig:wordle}
\end{figure}

The experiments have been run on an Intel i5 3.2GHz with 8GB RAM. The
\textsc{Random}, \textsc{Star Forest}, and \textsc{Cycle Cover}
algorithms finishes in under a second. The \textsc{CPDWCV} and
\textsc{Seam Carving} algorithms are based on a force-directed model,
and compute the word cloud with $100$ words within several seconds.

\end{document}